\documentclass[twocolumn,showpacs,preprintnumbers,prl]{revtex4}
\usepackage{amsmath}
\usepackage{amsthm}
\usepackage{graphicx}
\usepackage{dcolumn}
\usepackage{bm}

\newtheorem{theorem}{Theorem}

\begin{document}

\title{ Direct Measure of Quantum Correlation}
\author{Chang-shui Yu$^{1,2}$}
\email{quaninformation@sina.com; ycs@dlut.edu.cn}
\author{Haiqing Zhao$^{3}$}
\affiliation{$^{1}$H. H. Wills Physics Laboratory, University of Bristol, Tyndall Avenue,
Bristol, BS8 1TL, United Kingdom}
\affiliation{$^{2}$School of Physics and Optoelectronic Technology, Dalian University of
Technology, Dalian 116024, P. R. China}
\affiliation{$^{3}$School of Science, Dalian Jiaotong University , Dalian 116028, P. R.
China}
\date{\today }

\begin{abstract}
The quantumness of the correlation known as quantum correlation is usually
measured by quantum discord. So far various quantum discords can be roughly
understood as indirect measure by some special discrepancy of two
quantities. We present a direct measure of quantum correlation by revealing
the difference between the structures of classically and quantum correlated
states. Our measure explicitly includes the contributions of the
inseparability and local non-orthogonality of the eigenvectors of a density.
Besides its relatively easy computability, our measure can provide a unified
understanding of quantum correlation of all the present versions.
\end{abstract}

\maketitle

\textit{Introduction.-}As one kind of quantum correlation, quantum
entanglement has been playing an important role in quantum information
processing and become a necessary physical resource for quantum
communication and quantum computation [1]. However, a strong evidence shows
that there exists other nonclassical correlation-----quantum discord which
is shown to be effective in characterizing the computational efficiency of deterministic quantum
computation with one qubit [2-6]. Quantum discord captures the
fundamental feature of quantumness of nonclassical correlations, which is
much like quantum entanglement, but it is beyond quantum entanglement
because quantum discord is even present in separable quantum states [7].

Recently, quantum discord has attracted many interests in studying its
behavior under dynamical processes [8,9] and its operational meanings by
connecting it with Maxwell demon [10-12], or some quantum information processes such as broadcasting
of quantum states [10,11], quantum state merging [12,13], quantum
entanglement distillation [,14], entanglement of formation[15] and even the physical nature [16]. However,
as a measure, one can find that the definitions of quantum discord are
always based on indirect methods such as the original definition by the
discrepancy between quantum versions of two classically equivalent
expressions for mutual information and the latter distance-based definitions
[17,18] (including quantum mutual information or relative entropy which can
be considered as a special ''distance''). In addition, one can also find
that the two popular quantum discords are not symmetric if we exchange the
two subsystems, which even shows completely opposite behavior and that
quantum discords in terms of different definitions seem not to be consistent
with each other for some quantum states. In particular, if we compare
quantum correlation with quantum entanglement measure in the sense of that
quantum entanglement describes the inseparability of a quantum state in
mathematics [], a natural question is what the quantum correlation describes.
All above interesting properties require us to deepen our understanding not
only of different quantum discords but also from different angles. In this
paper, we attempt to answer the question mentioned above by giving a direct
definition of the measure of quantum correlation. Based on the features of
the structures of classically correlated states, we find our measure of
quantum correlation explicitly covers the contributions of inseparability
and local non-orthogonality of eigenvectors of a density. It can also
provide a unified understanding of quantum correlation corresponding to
various present quantum discords. In addition, it is obvious that our
measure can be relatively easily calculated.

\textit{Definitions of various classically correlated states.-} Considering
the various definitions of quantum correlation, in particular, the
asymmetric and the symmetric versions of quantum discord, one can
understand, in some angle, the difference comes from the different
definitions of classically correlated states. To our knowledge, fine
definitions of classically correlated states can be found in Ref. [10],
which can be generalized as follows. A multipartite quantum state $\rho
_{ABC\cdots }$ is semi-classically correlated, if it can be written as 
\begin{equation}
\rho _{ABC\cdots }=\sum p_{k}\left[ \otimes _{\alpha }\left| \psi
_{k}^{\alpha }\right\rangle \left\langle \psi _{k}^{\alpha }\right| \right]
\otimes \left[ \otimes _{\beta }\rho _{k}^{\beta }\right] ,
\end{equation}%
where $\left\{ \left| \psi _{k}^{\alpha }\right\rangle \right\} $ is some
orthonormal basis set of subsystem $\alpha $ and $\rho _{k}^{\beta }$ are
the quantum states of subsystem $\beta $, $\sum p_{k}=1,p_{k}>0$. Here we
omit the permutation of subsystems and $\alpha ,\beta =A,B,C,\cdots $. The
strictly classicially correlated states can be defined as 
\begin{equation}
\tilde{\rho}_{ABC\cdots }=\sum_{k}p_{k}\left[ \otimes _{\alpha }\left| \psi
_{k}^{\alpha }\right\rangle \left\langle \psi _{k}^{\alpha }\right| \right] ,
\end{equation}%
where $\left\{ \left| \psi _{k}^{\alpha }\right\rangle \right\} $ and $p_{k}$
are defined analogously.

From the quantum correlation point of view, we say that the quantum states
that cannot be written in Eq. (2) is quantum correlated and the quantum
states that cannot be given by either Eq. (1) or Eq. (2) is strictly quantum
correlated. In addition, if we only emphasize some subsystems, we can say
that Eq. (1) is quantum correlated subject to subsystems $\beta $ and
classically correlated subjected to subsystems $\alpha$.

\textit{Quantum correlation of general bipartite states.-} Just like that
the inseparability of a quantum state provides the direct contribution for
quantum entanglement, in order to give a direct measure of quantum
correlation, we have to find out what is the nature of the quantum
correlation. Let's first consider a general bipartite quantum state with the 
\textit{optimal eigendecomposition} as 
\begin{equation}
\rho _{AB}=\sum\limits_{i=1}^{n}\lambda _{i}\left| \phi _{i}\right\rangle
_{AB}\left\langle \phi _{i}\right|
\end{equation}%
with $n$ being the rank of $\rho _{AB}$. As we know, if a density matrix is
degenerate, the eigendecomposition is not unique. Suppose that the non-zero
eigenvalue $\lambda _{s}$, $0<s\leq n$, is degenerate with $f$ \ being the
degeneracy factor, the eigenvectors corresponding to one degenerate
eigenvalue can be transformed into another group of eigenvectors
corresponding to the same eigenvalue by the unitary transformation which can
be defined by%
\begin{equation}
\Phi U_{f}=\Phi ^{\prime },
\end{equation}%
where the columns of $\Phi $ and $\Phi ^{\prime }$ are all the eigenvectors
corresponding to the same degenerate eigenvalue and $U_{f}$ is an $f\times f$
unitary matrix. Therefore, the \textit{optimal eigendecomposition} as \{$%
\lambda _{i},\left| \phi _{i}\right\rangle $\} means the decomposition with
minimal average entanglement, i.e., $\sum\limits_{s}\frac{1}{f_{s}}%
\sum_{j=1}^{f_{s}}E\left( \left| \phi _{j}^{s}\right\rangle _{AB}\right)
=\min_{U_{f_{s}}}\sum\limits_{s}\frac{1}{f_{s}}\sum_{j=1}^{f_{s}}E\left(
U_{f_{s}}\left| \phi _{j}^{s}\right\rangle _{AB}\right) $. Note that the
optimal eigendecomposition is not unique either.

Comparing each eigenvector in Eq. (3) with that in Eq. (1) and Eq. (2), it
is obvious that the inseparability of the eigenvectors plays an important
role in quantum correlation. In addition, even though all the eigenvectors
are separable, the density could still has non-zero quantum correlation. Now
we can extract the reduced density of each eigenvector corresponding to
different subsystems. A further observation can reveal that the quantum
correlation will vanish, if for each subsystem, all the different (excluding
the same) reduced densities corresponding to different eigenvectors are pure
(which corresponds to the separability of eigenvectors) and orthogonal. The
latter orthogonality is the so called \textit{\ local non-orthogonality } we
mentioned in this paper. Therefore, we can draw a conclusion that the
quantum correlation includes two parts of contributions. One is the
inseparability of the eigenvector $\left\vert \phi _{i}\right\rangle _{AB}$,
and the other is the local non-orthogonality between the reduced densities
of $\left\vert \phi _{i}\right\rangle _{AB}$. We say that the two parts have
the equal contribution to quantum correlation in the sense that only one
part is not enough to distinguish the zero quantum correlation for general
quantum states. So a reasonable frame of the definition of quantum
correlation can be given by  
\begin{gather}
Q(\rho _{AB})=\sum_{i}\lambda _{i}E(\left\vert \phi _{i}\right\rangle _{AB})
\notag \\
+\min_{U_{f}}\left( \sum_{ij}w_{ij}F(\rho _{Ai}\rho _{Aj})+\sum_{ij}\tilde{w}%
_{ij}F(\rho _{Bi}\rho _{Bj})\right) ,
\end{gather}%
where i) $E(\left\vert \phi _{i}\right\rangle _{AB})$ \textit{is some
inseparability measure} which should obviously be given by a good
entanglement measure such as concurrence [19]. This implies that  $E(\left\vert \phi
_{i}\right\rangle _{AB}$ is not changed under local unitary transformations.
$F(\rho _{Ai}\rho _{Aj})$ should be some measure that captures the local
non-orthogonality and $w_{ij}$, $\tilde{w}_{ij}$ should be some weight-like
quantities. Since $F(\rho _{Ai}\rho _{Aj})$ is some measure of
non-orthogonality and  $\rho _{Ai}$ and $\rho _{Aj}$ belong to the same
subsystem,  $F(\rho _{Ai}\rho _{Aj})$ is also invariant under local unitary
operations.In detail, because we consider the eigendecomposition of $\rho
_{AB}$, we have to require ii) $F(\rho _{i}\rho _{j})=0$ for $\rho _{i}=\rho
_{j}$ or $\rho _{i}$\textit{\ orthogonal to }$\rho _{j}$\textit{,} otherwise%
\textit{\ }$F(\cdot )>0$. Thus so long as we can prove that $Q(\rho _{AB})$ is a
good criterion of quantum correlation, we can find that  $Q(\rho _{AB})$
satisfies all the necessary conditions for a quantum correlation measure [20],
which means that the quantum correlation measures are valid.
\begin{theorem}
$Q(\rho _{AB})$ defined in Eq. (5) measures the quantum correlation of $\rho
_{AB}$.
\end{theorem}

\begin{proof}
$Q(\rho _{AB})=0$ means that each term in Eq. (5) has zero value. So that the
first term vanishes implies that $\rho _{AB}=\sum_{ij}\tilde{\sigma}%
_{ij}\left| \tilde{\phi}_{i}\right\rangle _{A}\left\langle \tilde{\phi}%
_{i}\right| \otimes \left| \tilde{\psi}_{j}\right\rangle _{B}\left\langle 
\tilde{\psi}_{j}\right| $ with $\tilde{\sigma}_{ij}$ denoting the eigenvalue
and $\left| \tilde{\phi}_{i}\right\rangle \left| \tilde{\psi}%
_{j}\right\rangle $ corresponding to the eigenvector. That the last two
terms are zero means $\left| \tilde{\phi}_{i}\right\rangle $ and $\left| 
\tilde{\psi}_{j}\right\rangle $ are selected in an orthogonal set,
respectively. Therefore, we can draw a conclusion that the state $\rho _{AB}$
is strictly classically correlated. On the contrary, if $\rho _{AB}$ has no
quantum correlation, $\rho _{AB}$ can be written as $\rho
_{AB}=\sum_{ij}\sigma _{ij}\left| \phi _{i}\right\rangle _{A}\left\langle
\phi _{i}\right| \otimes \left| \psi _{j}\right\rangle _{B}\left\langle \psi
_{j}\right| $, where $\left| \phi _{i}\right\rangle _{A}$ and $\left| \psi
_{j}\right\rangle _{B}$ are chosen from some orthogonal set, respectively.
Thus $\left| \phi _{i}\right\rangle _{A}\left| \psi _{j}\right\rangle _{B}$
must be one group of eigenvectors of $\rho _{AB}$, which are obviously
separable. In other words, $\{\sigma _{ij},\left| \phi _{i}\right\rangle
_{A}\left| \psi _{j}\right\rangle _{B}\}$ is the optimal eigendecomposition.
Therefore, each term in Eq. (5) will vanish, which means $Q(\rho _{AB})=0$.
In addition, one can find that if $\rho _{AB}$ is a pure state, there exists
only one nonzero $\lambda _{i}$. Thus Eq. (5) will be reduced to $Q(\rho
_{AB})=E(\rho _{AB})$. This shows that our quantum correlation is equivalent
to entanglement for pure states.%
\end{proof}

From our Theorem 1, one can find that the quantum correlation is symmetric
if we exchange the two subsystems. However, how is it related to the
previous asymmetric quantum correlation? The answer can be easily found from
our theorem.

\begin{theorem}
Quantum correlation subject to system $A$ can be measured by

\begin{equation}
Q_{QC}(\rho _{AB})=\sum_{i=1}^{n}\lambda _{i}E(\rho
_{i})+\min_{U_{f}}\sum_{i,j=1}^{n}\lambda _{i}\lambda _{j}F(\rho _{Ai}\rho
_{Aj}),
\end{equation}%
and quantum correlation subject to system $B$ can be given by%
\begin{equation}
Q_{CQ}(\rho _{AB})=\sum_{i=1}^{n}\lambda _{i}E(\rho
_{i})+\min_{U_{f}}\sum_{i,j=1}^{n}\lambda _{i}\lambda _{j}F(\rho _{Bi}\rho
_{Bj}),
\end{equation}%
where all the parameters are defined the same as Theorem 1.
\end{theorem}

\begin{proof}
If $Q_{QC}(\rho _{AB})=0,$we can get the eigenvectors are separable
and $\rho _{Ai}$ are pure from the first term. We can get $\rho _{Ai}$ is
chosen from an orthogonal set from the second term. But it is not necessary
for $\rho _{Bi}$ to be orthogonal. So $\rho _{AB}$ must have the form of Eq.
(1). If $\rho _{AB}$ can be written as Eq. (1), one can quickly find that $%
Q_{QC}(\rho _{AB})=0$. The similar proof holds for Eq. (6).
\end{proof}

Our Theorem 1 shows an approach to distinguishing a quantum state with
quantum correlation from a strictly classical correlated one. An
intuintional observation can show that we can effectively distinguish a
strictly quantum correlated state from a semi-classically correlated state.
I.e.,

\begin{theorem}
The strictly quantum correlation can be measured by 
\begin{equation}
Q_{S}=\sqrt{Q_{CQ}Q_{QC}}.
\end{equation}
\end{theorem}

It is obvious that, $Q_{S}=0$ implies the quantum state is not strictly
quantum correlated and if a quantum state is semi-classically correlated, $%
Q_{S}=0$.

\textit{Multipartite quantum correlation.-}According to our analysis of
bipartite quantum correlations, one can find that our measures can be easily
generalized to multipartite quantum states. In order to simplify the
description, we first define $\mathcal{F}(\rho _{A},\rho _{B},\cdots
,0,\cdots )=\sum_{i,j=1}^{n}\lambda _{i}\lambda _{j}F(\rho _{Ai}\rho
_{Aj})+\sum_{k,l=1}^{n}\lambda _{k}\lambda _{l}F(\rho _{Bi}\rho
_{Bj})+\cdots $, where $\rho _{xi}=Tr_{\bar{x}}\left\vert \phi
_{i}\right\rangle \left\langle \phi _{i}\right\vert $ with $\lambda _{i}$
and $\left\vert \phi _{i}\right\rangle $ being the \textit{i}th eigenvalue
and eigenvector obtained from the optimal eigendecomposition of multipartite
quantum state $\rho _{ABC\cdots \text{ }}$and $\bar{x}$ denoting all the
subsystems except the subsystem $x$. It is obvious that each term on \textit{%
rhs.} corresponds to the variable $\rho _{xi}$ in $\mathcal{F}(\cdot )$ and
the variable $0$ in $\mathcal{F}(\cdot )$ means that there is no the
corresponding term on \textit{rhs.}. Thus we can claim what follows.

\begin{theorem}
The \textit{N}-partite quantum correlation of $\rho _{ABC\cdots \text{ }}$%
can be given by%
\begin{equation}
Q(\rho _{ABC\cdots \text{ }})=E(\rho _{ABC\cdots })+\min_{U_{f}}\mathcal{F}%
(\rho _{A},\rho _{B},\cdots ),
\end{equation}%
where $E(\cdot )$ is a multipartite entanglement measure with $E(\cdot )=0$
denoting fully separable states. The quantum correlation subject to some
subsystems can be defined by only preserving the corresponding reduced
densities in $\mathcal{F}(\cdot )$, i.e., 
\begin{equation}
Q_{x}(\rho _{ABC\cdots \text{ }})=E(\rho _{ABC\cdots })+\min_{U_{f}}\mathcal{%
F}(\rho _{x}).
\end{equation}%
Note $x$ in Eq. (10) can denote more than one subsystem due to the various
analogous definitions of quantum correlation. The strictly quantum
correlation can be defined as 
\begin{equation}
Q(\rho _{ABC\cdots \text{ }})=\sqrt[N]{\prod\limits_{x=A}^{N}Q_{x}(\rho
_{ABC\cdots \text{ }})}\text{.}
\end{equation}
\end{theorem}

The proof is analogous to those for bipartite quantum states.

\textit{Quantum correlation of two qubits as examples.-}The choice of the
functions $E\left( \cdot \right) $ and $F\left( \cdot \right) $ is not
unique. Based on the requirements of our quantum correlation measure, one
can let $E(\cdot )$ be the concurrence (the concurrence corresponding to all
possible bipartite grouping for a multipartite state), and $F(\rho _{i}\rho _{j})=2\bar{F}(1-\bar{F})$
where $\bar{F}(M)=TrM$ being
the fidelity with  $M=\sqrt{\sqrt{%
\rho _{i}}\rho _{j}\sqrt{\rho _{i}}}$. Alternatively, we can also define $E(\left\vert \phi
_{i}\right\rangle _{AB})=S(\rho _{R})$ with $\rho _{r}=Tr_{R}\left\vert \phi
_{i}\right\rangle _{AB}\left\langle \phi _{i}\right\vert $, $R=A$ or $B$ and 
$F(\rho _{i}\rho _{j})=2S(1-S)$ with $S(M)=-TrM\log _{2}M$. In both cases, one can set $w_{ij}=%
\tilde{w}_{ij}=\lambda _{i}\lambda _{j}$ with $\lambda _{i}$ denoting the
eigenvalue of $\rho _{AB}$ given in Eq. (3). Thus we can calculate the
various quantum correlations based on our definitions. However, from the
calculation point of view, one can find that it is quite easy to calculate
our quantum correlation, if a density matrix is not degenerate. On the
contrary, in general our special selection of degenerate eigenvectors
directly leads to the difficulty of the calculation. But we would like to
emphasize that our calculation is much simpler than the previous version of
quantum correlation [3,4] where we need to consider all potential Positive
Operator Value Measurements (POVM). In particular, if the degeneracy is
small, it is very possible to find an analytic expression for our quantum
correlation.

It is very interesting that for bipartite quantum states of qubits, the
above definitions of strictly quantum correlation can be reduced further. Now we can give our results in a rigid way.

\begin{theorem}
If $\rho _{AB}$ is a bipartite quantum state of qubits, then the strictly
quantum correlation can be measured by%
\begin{equation}
Q_{S}(\rho _{AB})=\sum_{i=1}^{n-g}\lambda _{i}C_{i}+\sum_{s=1}^{g}\lambda
_{s}f_{s}C(\rho _{s}),
\end{equation}%
where $\rho _{s}=\frac{1}{f_{s}}\sum_{i=1}^{f_{s}}\left\vert \phi
_{i}\right\rangle _{AB}\left\langle \phi _{i}\right\vert $ with $\left\vert
\phi _{i}\right\rangle _{AB}$ being the degenerate eigenvector subject to
the eigenvalue $\lambda _{s}$ and $C(\rho_s)$ means the concurrence of $\rho_s$. If the non-zero eigenvalues of $\rho _{AB}$
is $n$-fold degenerate with $n$ the rank of $\rho _{AB}$, then $Q_{S}(\rho
_{AB})=C(\rho _{AB}).$
\end{theorem}

\begin{proof} If $Q_{S}(\rho _{AB})=0$, then there exists one eigendecomposition
such that all the eigenvectors are separable. All the eigenvectors can be
written in the form of $\left\vert \alpha _{i}\right\rangle \left\vert \beta
_{i}\right\rangle $ with $\lambda _{i}$ the eigenvalues. Because $\rho _{AB}$
is $\left( 2\otimes 2\right) $-dimensional, it is impossible to find a third
vector $\left\vert \alpha _{i}\right\rangle $ ($\left\vert \beta
_{i}\right\rangle $) simultaneously orthogonal to two orthogonal vectors.
Thus at least $\left\vert \alpha _{i}\right\rangle $ or $\left\vert \beta
_{i}\right\rangle $ must belong to some orthogonal set. That is, $\rho
_{AB} $ can be written as one form of Eq. (2) or Eq. (1). Therefore, $%
\rho _{AB}$ has no strictly quantum correlation. On the contrary, if $\rho
_{AB}$ has no striclty quantum correlation, which means that $\rho _{AB}$ can
be in the form of Eq. (1) (Eq. (2) is included). It is obvious that one can always find
such an eigendecomposition that all the eigenvectors are separable.
Therefore, $Q_{S}(\rho _{AB})=0$. Note when $\rho _{AB}$ has degenerate
eigenvalues, the corresponding eigenvectors are not unique. The relation
between different eigenvectors are given in Eq. (4).  It is very interesting that Ref. [19]
shows that one can always find such a $U_{f}$ with $f$ being the degeneracy. The proof is completed.
\end{proof}

\begin{figure}[tbp]
\includegraphics[width=0.8\columnwidth]{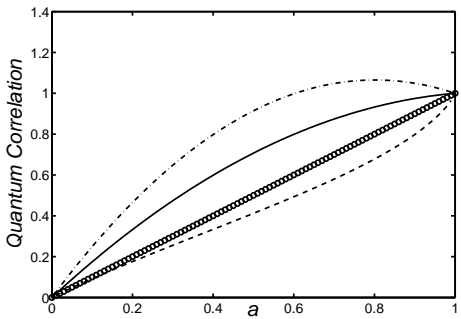}
\caption{(Dimensionless) Quantum correlations based on various definitions versus $a$. The dash-dotted line corresponds to the quantum correlation defined by Theorem 1;  The solid line corresponds to the strictly quantum correlation given by Theorem 3 and meanwhile it also corresponds to the quantum correlation subject A defined by Eq. (6); The "o" line corresponds to the simplified strictly quantum correlation given by Theorem 5; The dashed line is the quantum correlation subject to A given by Ref. [21]}
\label{1}
\end{figure}

As an explicit illustration of our quantum correlation measure, let's consider a class of states defined as
  $\rho=a\left\vert\psi^+\right\rangle\left\langle\psi^+\right\vert
  +(1-a)\left\vert11\right\rangle\left\langle 11\right\vert, 0\leq a\leq 1$, where $\left\vert\psi^+\right\rangle
  =(\left\vert 01\right\rangle+\left\vert 10\right\rangle)/\sqrt{2}$ is a maximally entangled state. The nonzero eigenvalues of $\rho$
  are given by $\lambda_1=a$ and $\lambda_2=1-a$ with the eigenvectors being $\left\vert\phi_1\right\rangle
  =\left\vert\psi^+\right\rangle$ and $\left\vert\phi_2\right\rangle=\left\vert 11\right \rangle$, respectively. It is 
  obvious that the eigenvalues are degenerate for $a=\frac{1}{2}$.  We employ concurrence as entanglement measure, so 
  one can find that for $a=\frac{1}{2}$, any decomposition of $\rho$ has the same average concurrence as $\frac{1}{2}$. In addition,
 it is very interesting that the fidelity $\bar{F}=\frac{\sqrt{2}}{2}$ for all valid decompositions. Therefore, we can easily calculate that the discord defined in Eq. (5) can be given by $\frac{1}{2}+\frac{1}{2}\cdot\frac{1}{2}\left(\frac{\sqrt{2}}{2}-\frac{1}{2}\right)\cdot 2\cdot 2=\frac{\sqrt{2}}{2}$ for $a=\frac{1}{2}$. Thus for all $0\leq a\leq 1$,  one can easily find its quantum correlation, which is plotted in FIG. 1. The quantum correlations with other definitions are also given in this figure. In particular, as a comparison we also plot the quantum correlation analytically given by Ref. [21].

\textit{Conclusions and discussion.-}We have defined a direct quantum
correlation measure which explicitly includes the contributions of the
inseparability and local non-orthogonality of the eigenvectors of a density.
Our measure can provide a unified understanding of quantum correlation of
all the present versions. Our definition can be reduced to quantum
entanglement for pure states, which is a fundamental property of quantum
correlation measure. It is shown that our quantum correlation measure of all
non-degenerate densities can be analytically calculative. In addition, we
have shown that for bipartite quantum states of two qubits, our definitions
can be simplified further by choosing different $F(\cdot )$. In this sense,
whether there exist other $F(\cdot )$ and $E(\cdot )$ which will lead to
simpler expressions deserves our further research. In particular, for
degenerate densities, how to select $F(\cdot )$ and $E(\cdot )$ such that
quantum correlation measure is analytically solved is a very interesting
question. In addition, from our definition of quantum correlation, one can
find that we start with the eigendecomposition of the density, which implies
we focus on the local non-orthogonality ahead of the inseparability. A
parallel consideration is to first address the inseparability. However,
whether the latter method can lead to another (or equivalent) group of
definitions of quantum correlation is another interesting question.

This work was supported by the National Natural Science Foundation of China,
under Grant No. 10805007 and No. 10875020, and the Doctoral Startup
Foundation of Liaoning Province.


\begin{thebibliography}{99}
\bibitem{[22]} Here the inseparability means to what degree a state, especially a pure state can not be factorized, which is mainly distinguished from the entanglement measure directly based on some kinds of distance, 
for example, relative entropy entanglement, distillation entanglement,
 entanglement cost, geometric entanglement and even negativity. However, concurrence seems different from those mentioned above, because at least, no definition based on a direct distance has been proposed for concurrence. In particular, for a mixed state, the concurrence can only be considered as the average contribution of the concurrence of the pure states that realizes the mixed state (convex roof extension).
\end{thebibliography}
\end{document}